\documentclass[draft]{elsart}

\usepackage[latin1]{inputenc}
\usepackage{alltt}
\usepackage{url}
\usepackage{latexsym}

\newtheorem{theorem}{Theorem}
\newtheorem{lemma}[theorem]{Lemma}
\newtheorem{conjecture}{Conjecture}

\newenvironment{proof}
               {\trivlist \item[\hskip\labelsep {\bf Proof.}]}
               {\hspace*{0pt plus 1fill}$\Box$\endtrivlist}
\newenvironment{constrproof}
               {\trivlist \item[\hskip\labelsep {\bf Proof.}]}
               {\hspace*{0pt plus 1fill}$\Box$\endtrivlist}
\newenvironment{classicproof}
               {\trivlist \item[\hskip\labelsep {\bf Proof (classical).}]}
               {\hspace*{0pt plus 1fill}$\Box$\endtrivlist}

\def\arraystretch{1}     

\def\irulebasic#1#2{            
   \begin{array}{@{}c@{}}       
     #1 \\[-1.2ex]              
     \hrulefill \\
     #2
   \end{array}}
\def\iruledoublebasic#1#2{      
   \begin{array}{@{}c@{}}       
     #1 \\[-1.2ex]              
     \hrulefill \\[-2.5ex]
     \hrulefill \\
     #2
   \end{array}}

\newenvironment{ruleset}
   {\begin{center}\lineskip=1.5ex plus 0.4ex minus 0.1ex} 
   {\end{center}}

\def\ruleleft{\hfilneg\hfil\hskip 1em} 
\def\ruleright{\hskip 1em\hfil}        

\def\srule#1{\ruleleft$#1$\ruleright}

\def\srulenumber#1#2{\ruleleft$#2$~~(#1)\ruleright}

\def\irule#1#2{\ruleleft$\irulebasic{#1}{#2}$\ruleright}
\def\irulenumber#1#2#3{\ruleleft$\irulebasic{#2}{#3}$~~(#1)\ruleright}

\def\iruledouble#1#2{\ruleleft$\iruledoublebasic{#1}{#2}$\ruleright}
\def\irulenumberdouble#1#2#3{\ruleleft$\iruledoublebasic{#2}{#3}$~~(#1)\ruleright}

\def\and{\hskip 1.5em\relax}

\newenvironment{syntaxleft}
   {\begin{flushleft}$\begin{array}{@{}l@{\quad}r@{~~}r@{~~}l@{\quad}l}}
   {\end{array}$\end{flushleft}}
\def\syntaxclass#1{\mbox{#1}}
\def\explanation#1{\mbox{#1}}
\def\alt{\mid}

\def\becomes{\leftarrow}
\def\fun{\rightarrow}

\def\dN{I \! \! N}              
\def\univ{{\cal U}}
\def\ccl{c}
\def\prem{{\cal A}}
\def\ensPart#1{\wp(#1)}
\def\sysInfer{\Phi}
\def\opInfer{F_\sysInfer}
\def\base{S}
\def\lfp#1{\mathrm{lfp}(#1)}
\def\gfp#1{\mathrm{gfp}(#1)}
\def\ensInd#1{\Delta(#1)}
\def\ensCoInd#1{\nabla(#1)}

\def\isvalue#1{#1 \in \mbox{\sf Values}}
\def\evalname{\Rightarrow}
\def\eval#1#2{#1 \evalname #2}
\def\evalinfname{\stackrel{\infty}{\Rightarrow}}
\def\evalinf#1{#1 \evalinfname {}}
\def\coevalname{\stackrel{\mbox{\sf\scriptsize co}}{\Rightarrow}}
\def\coeval#1#2{#1 \coevalname #2}
\def\redonename{\rightarrow}
\def\redone#1#2{#1 \redonename #2}
\def\redname{\stackrel{*}{\rightarrow}}
\def\red#1#2{#1 \redname #2}
\def\redinfname{\stackrel{\infty}{\rightarrow}}
\def\redinf#1{#1 \redinfname {}}
\def\coredname{\stackrel{\mbox{\sf\scriptsize co}*}{\rightarrow}}
\def\cored#1#2{#1 \coredname #2}
\def\implies{\Longrightarrow}
\def\land{\mathrel{\wedge}}
\def\lor{\mathrel{\vee}}
\def\lnot{\neg}
\def\notredname{\not\redone}
\def\notred#1{#1 \notredname {}}
\def\saferedname{\stackrel{\mbox{\sf\scriptsize safe}}{\rightarrow}}
\def\safered#1{#1 \saferedname {}}
\def\evalerr#1{#1 \evalname {\sf err}}
\def\compile#1{[\![#1]\!]}      

\def\redplusname{\stackrel{+}{\rightarrow}}
\def\redplus#1#2{#1 \redplusname #2}
\def\leftappheight#1{\| #1 \| }
\def\redinfNname#1{\mathop{\rightarrow}\limits^{\infty}_{#1}}
\def\redinfN#1#2{#2 \redinfNname{#1} {}}
\def\computename{{\cal C}}
\def\compute#1#2{\computename_{#1}(#2)}
\def\exec#1#2{{\cal D}(#1, #2)}
\def\error{\texttt{err}}
\def\defequal{\stackrel{\mbox{\scriptsize def}}{=}}
\def\bind{\mathrel{\rhd}}

\def\teval#1#2#3{#1 \evalname #3 \mathbin{/} #2}
\def\tevalinf#1#2{#1 \evalinfname {} \mathbin{/} #2}

\def\tred#1#2#3{#1 \redname #3 \mathbin{/} #2}
\def\tredinf#1#2{#1 \redinfname {} \mathbin{/} #2}
\def\bisim{\cong}
\def\redinfbisimname{\stackrel{\infty,\bisim}{\longrightarrow}}
\def\tredinfbisim#1#2{#1 \redinfbisimname {} \mathbin{/} #2}

\begin{document}

\begin{frontmatter}

\title{Coinductive big-step operational semantics}
\author[INRIA]{Xavier Leroy\corauthref{cor}}
\corauth[cor]{Corresponding author}
\ead{Xavier.Leroy@inria.fr}
\author[EMN]{Hervé Grall}
\address[INRIA]{INRIA Paris-Rocquencourt \\
                Domaine de Voluceau, B.P. 105,
                78153 Le Chesnay, France}
\address[EMN]{École des Mines de Nantes \\
              La Chantrerie, 4, rue Alfred Kastler, B.P. 20722,
              44307 Nantes, France}

\begin{abstract}
Using a call-by-value functional language as an example, this article
illustrates the use of coinductive definitions and proofs in big-step
operational semantics, enabling it to describe diverging evaluations
in addition to terminating evaluations.  We formalize the connections
between the coinductive big-step semantics and the standard small-step
semantics, proving that both semantics are equivalent.  We then study
the use of coinductive big-step semantics in proofs of type soundness
and proofs of semantic preservation for compilers.
A methodological originality of this paper is that all results
have been proved using the Coq proof assistant.  We explain the
proof-theoretic presentation of coinductive definitions and proofs
offered by Coq, and show that it facilitates the discovery and the
presentation of the results.
\end{abstract}


\end{frontmatter}

\section{Introduction}

There exist two widely-used styles of operational semantics:
{\em big-step semantics}, popularized by
Kahn \cite{Kahn-natural-semantics}
under the name {\em natural semantics}, relates programs to the final
results of their evaluations; {\em small-step semantics}, popularized
by Plotkin \cite{Plotkin-SOS-TR,Plotkin-SOS} under the name {\em
structural operational semantics}, repeatedly applies a one-step
reduction relation to form reduction sequences.
Small-step semantics is more expressive since it can describe the
evaluation of both terminating and non-terminating programs, as finite
or infinite reduction sequences, respectively.  In contrast, big-step
semantics describes only the evaluation of terminating programs, and
fails to distinguish between non-terminating programs and programs
that ``go wrong''.  For this reason, small-step semantics is generally
preferred, in particular for proving the soundness of type systems.

However, big-step semantics is more convenient than small-step
semantics for some applications.  One that is dear to our heart is
proving the correctness (preservation of program behaviours) of program
transformations, especially compilation of a high-level programming
language down to a lower-level language.  The first author's
experience and that of others
\cite{Leroy-compcert-06,Klein-Nipkow-jinja,Strecker-C0-05} is that
fairly complex, optimizing compilation passes can be proved correct
(for terminating source programs) relatively easily using
big-step semantics and inductions on the
structure of big-step evaluation derivations.  In contrast, compiler
correctness proofs using small-step semantics can address both
terminating and diverging source programs, but are more difficult
even for simple, non-optimizing compilation schemes
\cite{Hardin-Maranget-Pagano}.

In this article, we illustrate how coinductive definitions and
proofs enable big-step semantics to describe both terminating and diverging
evaluations.  The target of our study is a simple call-by-value
functional language.  We study two approaches: the first, initially
proposed by Cousot and Cousot \cite{Cousot-92}, complements the
normal inductive big-step evaluation rules for terminating evaluations with
coinductive big-step rules describing diverging evaluations; the
second simply interprets coinductively the normal big-step evaluation
rules, thus enabling them to describe both terminating and
non-terminating evaluations.  These semantics are defined in
sections~\ref{s-big-step}~and~\ref{s-coeval}, respectively.  

The main technical results of this article are of two kinds.  First,
we prove that the coinductive big-step definition of divergence is
equivalent to the more familiar definitions using either small-step
semantics (section~\ref{s-small-step}) or a simple form
of denotational semantics (section~\ref{s-den-sem}).  We also extend
these equivalence results to trace semantics (section~\ref{s-traces}).
Then, we study two applications of the big-step definition of
divergence: a novel approach to stating and proving the soundness of
type systems (section~\ref{s-type-soundness}), and proofs of semantic
preservation for compilation down to an abstract machine
(section~\ref{s-compiler-correctness}).

An originality of this article is that all results were not only proved
using a proof assistant (the Coq system), but even developed in
interaction with this tool, and only then transcribed to standard
mathematical notations.
The Coq proof assistant \cite{Coq,Bertot-Casteran-Coqart} provides
built-in support for coinductive definitions and proofs by
coinduction.  This support follows a proof-theoretic approach to
induction and coinduction that we present in section~\ref{s-methodology}
and relate with the standard approach using fixed points.  The
proof-theoretic approach leads to proofs by coinduction that are
simpler than the standard arguments based on $F$-consistent relations
\cite{Milner-Tofte-coinduction,Gapeyev-Levin-Pierce}.
Our use of Coq has therefore been
doubly beneficial: it facilitated the discovery and presentation of
the results in this article, while at the same time generating strong
confidence in them.

\section{Induction and coinduction: A proof-theoretic approach} 
\label{s-methodology}

Following the classical presentation of Aczel \cite{Aczel1977}, 
an inference system over a set $\univ$ of judgments 
is a set of inference rules.
An \emph{inference rule} is an ordered pair 
$(\prem,\ccl)$, where $\ccl\in \univ$ is the  
\emph{conclusion} of the rule and
$\prem\subseteq \univ$ is the set of its
\emph{premises} or \emph{antecedents}.
A rule is usually written as follows:
$$\irulebasic{
    ~~\prem~~
}{
    \ccl
}$$
The intuitive interpretation of this rule is that the judgment $\ccl$
can be inferred from the set of judgments $\prem$.

\subsection{Fixed-point approach}
One way to give meaning to an inference system is to consider the
fixed points of the associated inference operator.  If $\sysInfer$ is
an inference system over $\univ$, we define the operator
$\opInfer : \ensPart{\univ} \rightarrow \ensPart{\univ}$ as
$$ \opInfer(\base) = \{ \ccl \in \univ \mid
    \exists \prem \subseteq \base,~ (\prem,\ccl) \in \sysInfer \}. $$
In other terms, $\opInfer(\base)$ is the set of judgments that can be
inferred in one step from the judgments in $\base$ by using the
inference rules.

A set $\base$ is said to be
\emph{closed} if $\opInfer(\base) \subseteq \base$,
and \emph{consistent} if $\base \subseteq \opInfer(\base)$.  A closed set
$\base$ is such that no new judgments can be inferred from $\base$.  A
consistent set~$\base$ is such that all judgments that cannot be inferred
from $\base$ are not in $\base$.  

The inference operator is monotone:
$\opInfer(\base) \subseteq \opInfer(\base')$ if $\base \subseteq \base'$.
By Tarski's fixed point theorem for complete lattices
\cite[p. 286]{Tarski1955}, it follows that the inference operator
possesses both a least fixed point and a greatest fixed point, which
are the smallest $\opInfer$-closed set and the largest
$\opInfer$-consistent set, respectively.
\begin{eqnarray*}
\lfp{\opInfer} & = &
    \bigcap ~ \{ \base \mid \opInfer(\base) \subseteq \base \} \\
\gfp{\opInfer} & = &
    \bigcup ~ \{ \base \mid \base \subseteq \opInfer(\base) \}
\end{eqnarray*}
The least fixed point $\lfp{\opInfer}$ is the \emph{inductive interpretation}
of the inference system $\sysInfer$, and the greatest fixed point
$\gfp{\opInfer}$ is its \emph{coinductive interpretation}.
These interpretations lead to the following two proof principles:
\begin{itemize}
\item Induction principle: to prove that all judgments in the
  inductive interpretation belong to a set $\base$, show that
  $\base$ is $\opInfer$-closed.
\item Coinduction principle: to prove that all judgments in a set $\base$
  belong to the coinductive interpretation, show that $\base$
  is $\opInfer$-consistent.
\end{itemize}

\subsection{Proof-theoretic approach} \label{s-proof-approach}
In contrast with the fixed point approach,
the proof-theoretic approach starts from the \emph{proofs}
admissible in an inference system.  These proofs naturally correspond
to \emph{derivations}, also called \emph{proof trees}.  These are
trees whose nodes are labeled with judgments $\ccl \in \univ$ and such
that for all nodes~$n$, the label~$\ccl$ of~$n$ and the labels~$\prem$
of the children of~$n$ correspond to an inference rule: $(\prem, \ccl) \in
\sysInfer$.  The conclusion of a derivation is the label of its root node.

A derivation $d$ is \emph{well-founded} if it has no infinite branch;
$d$ is \emph{ill-founded} otherwise.  If every rule in $\sysInfer$ has a
finite set of premises, well-founded derivations are finite while
ill-founded derivations are infinite.  

In the proof-theoretic approach, the \emph{inductive interpretation}
of the inference system $\sysInfer$ is the set $\ensInd{\sysInfer}$
of conclusions of well-founded derivations, while the
\emph{coinductive interpretation} is the set $\ensCoInd{\sysInfer}$
of conclusions of arbitrary derivations (ill-founded or well-founded).
These interpretations come with the following proof principles:
\begin{itemize}
\item Induction principle: to prove that all judgments in the
  inductive interpretation belong to a set $\base$, proceed by
  structural induction over well-founded derivations.  That is,
  show that $\ccl \in \base$ if $\ccl$ is the conclusion of a
  derivation $d$, assuming that $j \in \base$ for all
  conclusions $j$ of the strict subderivations of $d$.
\item Coinduction principle: to prove that all judgments in a set $\base$
  are in the coinductive interpretation, build a system
  of recursive equations between derivations, with unknowns $(x_j)_{j \in \base}$.
  Each equation is of the form
$$
    x_j = \irulebasic{
             x_{j_1} \and x_{j_2} \and \ldots
}{
             j
          }
$$
and must be justified by an inference rule:
$(\{j_1, j_2, \ldots\}, j) \in \sysInfer$.  These equations are
\emph{guarded}, meaning that there are no trivial equations $x_j = x_{j'}$.
It follows that the system has a unique solution \cite{Courcelle1979a},
and this solution $\sigma$ is such that for all $j \in \base$,
$\sigma(x_j)$ is a valid derivation that proves~$j$.  Therefore,
all $j \in \base$ are also in $\ensCoInd{\sysInfer}$.
\end{itemize}

\subsection{Equivalence between the two approaches}
The following theorem shows that the interpretations defined using
fixed points and using derivations coincide.

\begin{theorem}
For all inference systems $\sysInfer$,
$\lfp{\opInfer} = \ensInd{\sysInfer}$ and
$\gfp{\opInfer} = \ensCoInd{\sysInfer}$.
\end{theorem}

\begin{proof} 
It is easy to show that $\ensInd{\sysInfer}$ is $\opInfer$-closed
and that $\ensCoInd{\sysInfer}$ is $\opInfer$-consistent.
Therefore, $\lfp{\opInfer} \subseteq \ensInd{\sysInfer}$
and $\ensCoInd{\sysInfer} \subseteq \gfp{\opInfer}$.

Consider a $\opInfer$-closed set $\base$.  A structural induction over
well-founded derivations $d$ shows that the conclusion of $d$ is in
$\base$.  Therefore, $\ensInd{\sysInfer} \subseteq S$.
Since $\lfp{\opInfer}$ is $\opInfer$-closed, the inclusion
$\ensInd{\sysInfer} \subseteq \lfp{\opInfer}$ follows.

Finally, consider a $\opInfer$-consistent set $\base$.  For any judgment
$j$ in $\base$, there exists a rule $( K_j, j)$ in $\sysInfer$,
where $K_j \subseteq \base$.  
We define a system of guarded recursive equations,
with variables $(x_j)_{j\in \base}$.
$$
    x_j = \irulebasic{
             (x_k)_{k \in K_j}
}{
             j
          }
$$
The solution $\sigma$ of this system is such that for all $j \in \base$,
the derivation $\sigma(x_j)$ is valid in $\sysInfer$ and
proves $j$.  Therefore, $\base \subseteq \ensCoInd{\sysInfer}$.
Since $\gfp{\opInfer}$ is $\opInfer$-consistent, the inclusion
$\gfp{\opInfer} \subseteq \ensCoInd{\sysInfer}$ follows.
\end{proof}

The equality $\lfp{\opInfer} = \ensInd{\sysInfer}$ is proved by 
Aczel \cite{Aczel1977}.  The equality
$\gfp{\opInfer} = \ensCoInd{\sysInfer}$ is proved in
the second author's PhD dissertation \cite[p.~77]{Grall-phd},
but to our knowledge there is no other published proof.
This is, however, a well-known result.  For instance, it has recently
been used to extend logic programming with coinductive terms and
derivations \cite{Simon-Mallya-Bansal-Gupta-06}.

\subsection{Induction and coinduction in the Coq proof assistant}
The Coq proof assistant that we use to
develop the present work follows the proof-theoretic formulation
of induction and coinduction.  In accordance with the
propositions-as-types, proofs-as-programs paradigm, inference systems
are presented as inductively or coinductively-defined predicates,
resembling data type definitions in ML or Haskell.  Such a predicate
is defined by a set of constructors, corresponding to 
inference rules.  Applied to terms representing proofs for its
premises, a constructor returns a proof term for its conclusion.

Proofs by induction and by coinduction are both represented as
recursive functions.  For a proof by induction, the Coq type system
demands that the recursive function
be \emph{structural}: the arguments to recursive calls are strict
subterms of the recursive parameter.  For a proof by coinduction, the
Coq type system demands that the recursive function be \emph{productive}:
its result is a constructor application, and the results of recursive
calls are only used as arguments to this constructor.  Such
productive recursive functions correspond closely to the systems of guarded
equations used above.

While proof terms can be provided explicitly by the user, most of the
time they are built incrementally by the Coq proof assistant in
response to tactics entered by the user.  When using tactics, proofs
by coinduction are as easy to conduct as proofs by induction: in
response to the {\tt cofix} tactic, the system provides the expected
result as an additional hypothesis, then makes sure that this
hypothesis is only used in positions permitted by productive recursive
functions.  (See \cite{Gimenez-coind-94} and
\cite[chap. 13]{Bertot-Casteran-Coqart} for more details, and the
proof of lemma~\ref{evalinf-omega} below for a concrete example.)
The proof sketches we give in the remainder of this article are
written in the same proof style, and play fast and loose with
coinduction.  In particular, except for the very first proofs, we do
not exhibit $\opInfer$-consistent sets nor systems of guarded
equations between derivations.  The skeptical reader is referred to
the corresponding Coq development \cite{Leroy-Grall-coindsem-Coq} for full
details.

Coq is based on a constructive logic (the Calculus of Constructions),
but proofs in classical logic can be expressed in Coq by adding
axioms that are known to be consistent with Coq's logic.
The majority of our proofs are constructive, but some use the axiom
of excluded middle.  The proofs that use this axiom are marked 
``{\bf (classical)}''.

\section{The language and its big-step semantics} \label{s-big-step}

The language we consider in this article is the $\lambda$-calculus
extended with constants: the simplest functional language
that exhibits run-time errors (terms that ``go wrong'').  Its syntax
is as follows:
\begin{syntaxleft}
\syntaxclass{Variables:} &
x, y, z, \ldots
\\
\syntaxclass{Constants:} &
c & ::= & 0 \alt 1 \alt \ldots
\\
\syntaxclass{Terms:} &
a, b, v & ::= & x \alt c \alt \lambda x. a \alt a~b
\end{syntaxleft}
We write $a[x \becomes b]$ for the capture-avoiding substitution\footnote{%
The Coq development does not treat terms modulo $\alpha$-conversion,
therefore the substitution $a[x \becomes b]$ can capture variables.
However, it is capture-avoiding if $b$ is closed, and this
suffices to define evaluation and reduction of closed source terms.}
of $b$ for all free occurrences of $x$ in $a$.  We say that a term $v$ is
a value, and write $\isvalue v$, if $v$ is either a constant $c$ or an
abstraction $\lambda x. b$.

The standard call-by-value semantics in big-step style for this
language is defined by the inductive interpretation of the following
inference rules.  They define the relation $\eval a v$ (read: ``$a$
evaluates to $v$'').
\begin{ruleset}
\srulenumber{$\evalname$-const}{  \eval c c   }
\srulenumber{$\evalname$-fun}{    \eval {\lambda x.a} {\lambda x.a}  }
\irulenumber{$\evalname$-app}{
    \eval {a_1} {\lambda x.b} \and
    \eval {a_2} {v_2} \and
    \eval {b[x \becomes v_2]} v
}{
    \eval {a_1~a_2} v
}
\end{ruleset}

\begin{lemma} \label{eval-isvalue}
If $\eval a v$, then $\isvalue v$.
\end{lemma}

\begin{constrproof} Induction on a derivation of $\eval a v$.
\end{constrproof}

\begin{lemma} \label{eval-deterministic}
The $\evalname$ relation is deterministic: 
if $\eval a v$ and $\eval a {v'}$, then $v = v'$.
\end{lemma}

\begin{constrproof} By induction on the derivation of $\eval a v$
and case analysis over that of $\eval a {v'}$.
\end{constrproof}

The rules above capture only terminating evaluations.  Writing $\delta
= \lambda x.~x~x$ and $\omega = \delta~\delta$, we have for instance:

\begin{lemma} $\eval \omega v$ is false for all terms $v$. \end{lemma}
\begin{constrproof} We show that $\eval a v$ implies $a \not= \omega$
by induction on the derivation of $\eval a v$. \end{constrproof}

Following Cousot and Cousot \cite{Cousot-92} and
the second author's PhD work
\cite{Grall-phd}, we define divergence (infinite evaluations) by the
coinductive interpretation\footnote{%
Throughout this article, double horizontal lines in inference rules
denote inference rules that are to be interpreted coinductively;
single horizontal lines denote the inductive interpretation.}
of the following inference rules.  They define the relation
$\evalinf a$ (read: ``$a$ diverges'').

\begin{ruleset}
\irulenumberdouble{$\evalinfname$-app-l}{
         \evalinf {a_1}
}{
         \evalinf {a_1~a_2}
}
\irulenumberdouble{$\evalinfname$-app-r}{
         \eval {a_1} v \and \evalinf {a_2}
}{
         \evalinf {a_1~a_2}
}
\irulenumberdouble{$\evalinfname$-app-f}{
         \eval {a_1} {\lambda x.b} \and \eval {a_2} v \and 
         \evalinf {b[x \becomes v]}
}{
         \evalinf {a_1~a_2}
}
\end{ruleset}
Note that we have imposed (arbitrarily) a left-to-right evaluation
order for applications.

\begin{lemma} \label{evalinf-omega}
$\evalinf \omega$ holds.
\end{lemma}

\begin{constrproof} 
The proof is by coinduction.  Assume $\evalinf \omega$ as
coinduction hypothesis.  We can derive $\evalinf \omega$ with rule
($\evalinfname$-app-f), using the coinduction hypothesis as third
premise.

Since this is the first proof by coinduction in this article, we now
detail the proof sketch given above using the various approaches
outlined in section~\ref{s-methodology}.

{\it Greatest fixed point.}
Consider the inference operator $F$
associated with the rules defining $\evalinfname$, namely
$$ F(S) = \begin{array}[t]{ll}
     & \{ a_1~a_2 \mid a_1 \in S \} \\
\cup & \{ a_1~a_2 \mid \exists v,~ \eval {a_1} v \wedge a_2 \in S \} \\
\cup & \{ a_1~a_2 \mid \exists x,b,v,~
            \eval {a_1} {\lambda x.b} \wedge
            \eval {a_2} v \wedge
            b[x \becomes v] \in S \}
\end{array}$$
The set $S = \{ \omega \}$ is $F$-consistent.  Indeed, $\omega \in
F(\{\omega\})$ by the third line of the definition of $F$.  Therefore,
$S \subseteq \gfp{F}$, implying that $\evalinf \omega$ holds.

{\it Systems of guarded recursive equations.}
Consider the following equation with unknown $d$ (a derivation):
$$
    d = \irulebasic{
             \eval \delta {\lambda x.~x~x} \and
             \eval \delta \delta \and
             d
}{
             \evalinf {\delta~\delta}
        }
$$
Since $(x~x)[x \becomes \delta] = \delta~\delta$, this equation is
justified by rule ($\evalinfname$-app-f).  Moreover, it is guarded.
Therefore, its solution is a valid derivation that proves
$\evalinf {\delta~\delta}$.  It follows that this judgment holds.

{\it Coq proof term.}
Consider the Coq proof term
{\tt evalinf{\char95}omega} defined by the following corecursion:
\begin{verbatim}
CoFixpoint evalinf_omega : evalinf omega :=
  let eval_delta : eval delta delta :=
    eval_fun x (App (Var x) (Var x)) in 
  evalinf_app_f delta delta x (App (Var x) (Var x)) delta
    eval_delta
    eval_delta
    evalinf_omega.
\end{verbatim}
The two constructor functions {\tt eval{\char95}fun} and {\tt evalinf{\char95}app{\char95}f}
correspond to the inference rules ($\evalname$-fun) and
($\evalinfname$-app-f), respectively.  They receive as arguments
instantiations for the free variables of the rules ($x$ and $a$ for 
($\evalname$-fun); $a_1$, $a_2$, $x$, $b$, $v$ for
($\evalinfname$-app-f)),
followed by proof terms for their premises 
(proofs of $\eval \delta \delta$, $\eval \delta \delta$ and
$\evalinf \omega$ for ($\evalinfname$-app-f)).
The term {\tt evalinf{\char95}omega} has type {\tt evalinf\ omega}, which proves that this
proposition representing $\evalinf \omega$ is true.

{\it Coq proof script.} The following commented sequence of
tactics builds the proof term above in an interactive manner.
\def\comment#1{\hspace*{2cm}{\rm\it#1}}
\begin{alltt}
Lemma evalinf_omega: evalinf omega.
Proof.
  cofix COINDHYP.
\comment{Prepare a proof by coinduction.  The current goal \(\evalinf\omega\)}
\comment{becomes an hypothesis named {\tt{COINDHYP}}}
  unfold omega. eapply evalinf_app_f.
\comment{Apply the constructor for rule \(\evalinfname\)-app-f}
  unfold delta. apply eval_fun. 
\comment{Prove the first premise (evaluation of \(\delta\))}
  unfold delta. apply eval_fun.
\comment{Prove the second premise (evaluation of \(\delta\))}
  simpl. fold delta. fold omega.
\comment{Replace \((x x)[x\becomes\delta]\) by \(\omega\).}
  apply COINDHYP.
\comment{Prove the third premise by invoking the coinduction hypothesis.}
Qed.
\end{alltt}
\end{constrproof}

\begin{lemma} $\eval a v$ and $\evalinf a$ are mutually exclusive.
\end{lemma}
\begin{constrproof} By induction on the derivation of $\eval a v$,
case analysis on that of $\evalinf a$, and lemma~\ref{eval-deterministic}.
\end{constrproof}

Programs that neither evaluate nor diverge according to the rules
above are said to ``go wrong''.  For instance, the program $0~0$ goes
wrong since neither $\eval {0~0} v$ nor $\evalinf {0~0}$ hold for any $v$.

\section{Relation with small-step semantics} \label{s-small-step}

The one-step reduction relation $\redonename$ is defined by the
call-by-value $\beta$-reduction axiom plus two context rules for
reducing under applications, assuming left-to-right evaluation order.

\begin{ruleset}
\irulenumber{$\redonename$-$\beta$}{
   \isvalue v
}{
   \redone{(\lambda x.a)~v} {a[x \becomes v]}
}
\\
\irulenumber{$\redonename$-app-l}{
   \redone {a_1} {a_2}
}{
   \redone {a_1~b} {a_2~b}
}
\irulenumber{$\redonename$-app-r}{
   \isvalue a \and \redone {b_1} {b_2}
}{
   \redone {a~b_1} {a~b_2}
}
\end{ruleset}

\begin{lemma} \label{red1-determ}
The $\redonename$ relation is deterministic: 
if $\redone a {a'}$ and $\redone a {a''}$, then $a' = a''$.
\end{lemma}

\begin{constrproof} By induction on the derivation of $\redone a {a'}$
and case analysis over that of $\redone a {a''}$.
\end{constrproof}

There are three kinds of reduction sequences of interest.
The first, written $\red a b$ (``$a$ reduces to $b$ in zero, one or
several steps''), is the standard reflexive transitive closure of
$\redonename$; it captures finite reductions.  The second, written $\redinf a$
(``$a$ reduces infinitely''), captures infinite reductions.  The third,
written $\cored a b$ (``$a$ reduces to $b$ in zero, one, several or infinitely
many steps''), is the coinductive interpretation of the rules for
reflexive transitive closure; it captures both finite and infinite
reductions.  These relations are defined by the following rules,
interpreted inductively for $\redname$ and coinductively for
$\redinfname$ and $\coredname$.
\begin{ruleset}
\srule{\red a a }
\hspace*{4.5cm}
\srule{\cored a a }
\\
\irule{
       \redone a {a'} \and \red {a'} {b}
}{
       \red a {b}
}
\iruledouble{
       \redone a {a'} \and \redinf {a'}
}{
       \redinf a
}
\iruledouble{
       \redone a {a'} \and \cored {a'} {b}
}{
       \cored a {b}
}
\end{ruleset}

It is true that $\coredname$ is the union of $\redname$ and
$\redinfname$, in the following sense.

\begin{lemma} \label{cored-red-or-redinf}
 $\cored a b$ if and only if $\red a b$ or $\redinf a$.
\end{lemma}

\begin{classicproof} For the ``if'' part, we show that
$\red a b \implies \cored a b$ by induction on $\red a b$,
and that $\redinf a \implies \cored a b$ by coinduction.
For the ``only if'' part, we show that $\cored a b \land \lnot (\red a b)
\implies \redinf a$ by coinduction.  The result follows
by excluded middle over $\red a b$.
\end{classicproof}

We now turn to relating the reduction relations (small-step) and the
evaluation relations (big-step).  It is well known that normal
evaluation is equivalent to finite reduction to a value.

\begin{theorem} \label{eval-red-equiv}
$\eval a v$ if and only if $\red a v$ and $\isvalue v$.
\end{theorem}

\begin{constrproof}  The ``only if'' part is an easy induction on
$\eval a v$.  For the ``if'' part, we first show the following two
lemmas: (1) $\eval v v$ if $\isvalue v$, and (2) $\eval a v$ if 
$\redone a b$ and $\eval b v$.  The result follows by induction on the
proof of $\red a v$. \end{constrproof}

Similarly, divergence ($\evalinfname$) is equivalent to infinite
reduction ($\redinfname$).  The proof uses the following lemma.

\begin{lemma} \label{red-or-redinf}
For all terms $a$, either $\redinf a$, or there exists $b$ such that
$\red a b$ and $\notred b$, that is, $\forall b', ~\lnot(\redone b {b'})$.
\end{lemma}
\begin{classicproof}  We first show that
$\forall b,~ \red a b \implies \exists b', ~\redone b {b'}$ implies
$\redinf a$ by coinduction.  We then argue by excluded
  middle on $\redinf a$.
\end{classicproof}

\begin{theorem} \label{evalinf-redinf-equiv}
 $\evalinf a$ if and only if $\redinf a$. 
\end{theorem}

\begin{classicproof}  For the ``only if'' part, we first show that
$\evalinf a$ implies $\exists b,~ \redone a b \land \evalinf b$ by
structural induction on $a$, then conclude by
coinduction.  For the ``if'' part, we
proceed by coinduction and case analysis over $a$.  The only
non-trivial case is $a = a_1~a_2$.  Using lemma~\ref{red-or-redinf},
we distinguish three cases: (1) $a_1$ reduces infinitely; (2) $a_1$
reduces to a value but $a_2$ reduces infinitely; (3) $a_1$ and $a_2$ reduce
to values $\lambda x.b$ and $v$ respectively, and $b[x \becomes v]$
reduces infinitely.  We conclude $\evalinf a$ by applying the
appropriate inference rule  for each case, the coinduction hypothesis for the
$\evalinfname$ premise, and theorem~\ref{eval-red-equiv} for the
$\evalname$ premises.
\end{classicproof}

\section{Relation with denotational semantics} \label{s-den-sem}

Denotational semantics is an alternate way to characterize divergent
and convergent terms.  In this section, we develop a simple
denotational semantics for call-by-value $\lambda$-calculus and prove
that it captures the same notions of convergence and divergence as our
big-step operational semantics.  To facilitate the mechanization of
these results in the Coq theorem prover, we adopt an elementary
presentation of the denotational semantics that does not require the
full generality of Scott domains.

We define the computation $\compute n a$ of a term $a$ at maximal
recursion depth $n \in \dN$ by recursion over $n$, as follows.
\begin{eqnarray*}
\compute 0 a & = & \bot \\
\compute {n+1} {x} & = & \error \\
\compute {n+1} {c} & = & c \\
\compute {n+1} {\lambda x.a} & = & \lambda x.a \\
\compute {n+1} {a_1~a_2}
& = & \compute n {a_1} \bind (v_1 \mapsto \\
&   & ~~\compute n {a_2} \bind (v_2 \mapsto \\
&   & ~~~~{\tt if\ } v_1 = \lambda x.b {\tt \ then\ }\compute{n}{b[x \becomes v_2]} {\tt \ else\ }\error))
\end{eqnarray*}
The monadic composition operator $\bind$ used in the application case is
defined by
$$ \bot \bind f = \bot \qquad
   \error \bind f = \error \qquad
   v \bind f = f(v). $$
The result of $\compute n a$, or in other terms the outcome of
executing $a$ at depth $n$, is one of the following three possibilities:
(1) a value~$v$, denoting normal termination with $v$ as final value;
(2) the symbol $\error$, denoting abrupt termination on a run-time
  error (such as encountering a free variable or an application of a
  constant);
(3) the symbol $\bot$, indicating that the computation cannot
  complete within $n$ recursive steps.

The flat ordering $\le$ over results is defined by $\bot \le r$ and 
$r \le r$ for all~$r$.  The $\computename$ function is monotone with
respect to this ordering:

\begin{lemma} \label{compute-incr}
If $n \le m$, then $\compute n a \le \compute m a$.
\end{lemma}

\begin{constrproof} By induction over $n$ and case analysis over $a$.
\end{constrproof}

We say that a term $a$ executes with result $r$, or in other terms
that $r$ is the denotation of $a$, and we write $\exec a r$, if
$\compute n a = r$ for almost all $n$:
$$ \exec a r ~ \defequal ~
   \exists p,~ \forall n,~n \ge p ~\implies~ \compute n a = r. $$
Since $\computename$ is monotone, the following properties hold trivially:

\begin{lemma}\label{exec-either}
If $\exec a r$, then for all $n$, either $\compute n a = \bot$
or $\compute n a = r$.
\end{lemma}

\begin{lemma}\label{exec-notbot}
If $r \not= \bot$ and $\compute n a = r$ for some $n$, then $\exec a r$.
\end{lemma}

\begin{lemma}\label{exec-bot}
$\exec a \bot$ if and only if $\compute n a = \bot$ for all $n$.
\end{lemma}

It follows that every term has one and exactly one denotation.

\begin{lemma} \label{evaluates-total}
For all terms $a$, there exists a result $r$ such that
$\exec a r$.
\end{lemma}

\begin{classicproof} By excluded middle, either
$\forall n, \compute n a = \bot$ or $\exists n, \compute n a \not= \bot$.
In the former case, we obviously have $\exec a \bot$.  In the latter case,
pick $n$ such that $\compute n a \not= \bot$ and take $r = \compute n a$.  By
lemma~\ref{exec-notbot}, we have $\exec a r$.
\end{classicproof}

\begin{lemma} If $\exec a {r_1}$ and $\exec a {r_2}$, then $r_1 = r_2$.
\end{lemma}

\begin{constrproof} 
Notice that $r_1 = \compute n a = r_2$ for sufficiently large $n$.
\end{constrproof}

We now relate this denotational semantics with the big-step
operational semantics of section~\ref{s-big-step}, starting with the
terminating case.

\begin{theorem} \label{eval-converges}
$\eval a v$ if and only if $\exec a v$.
\end{theorem}

\begin{constrproof} For the ``if'' part, we show that
$\compute n a = v$ implies $\eval a v$ by induction over $n$ and
case analysis over $a$ and over the results of the recursive computations.
The case $a = x$ contradicts the hypothesis $\compute n a = v$.
For the cases $a = c$ or $a = \lambda x.b$, we have $v = a$ by
definition of $\computename$ and the result follows by 
rules ($\evalname$-const) or ($\evalname$-fun).
Finally, if $a = a_1~a_2$, the exploitation of the hypothesis
$\compute n a = v$ leads to $\compute {n-1} {a_1} = \lambda x.b$
and $\compute {n-1} {a_2} = v_2$ and
$\compute {n-1} {b[x \becomes v_2]} = v$.  The result follows from the
induction hypothesis and rule ($\evalname$-app).

For the ``only if'' part, we proceed by induction over the
derivation of $\eval a v$ and exhibit an $n$ such that $\compute n a = v$.
From this, $\exec a v$ follows by lemma~\ref{exec-notbot}.
The cases where $a$ is a constant or a
function are trivial, since $\compute 1 a = v$ in these cases.
For the application case $a = a_1~a_2$, the induction hypothesis leads to
$\compute {n_1} {a_1} = {\lambda x.b}$ and
$\compute {n_2} {a_2} = {v_2}$ and
$\compute {n_3} {b[x \becomes v_2]} = v$ for some $n_1, n_2, n_3$.
Taking $n = 1 + \max(n_1,n_2,n_3)$, we have $\compute n a = v$ by definition
and monotonicity of $\computename$, and the result follows.
\end{constrproof}

\begin{theorem} $\evalinf a$ if and only if $\exec a \bot$.  \end{theorem}

\begin{constrproof} For the ``only if'' part, we show that
$\evalinf a$ implies $\compute n a = \bot$
by induction over $n$ and case analysis on the last rule used in the
derivation of $\evalinf a$.  In all three cases, $a = a_1~a_2$.
If $\evalinf {a_1}$, $\compute n a = \compute {n-1} {a_1} = \bot$ by
induction hypothesis.  If $\eval {a_1} {v_1}$ and $\evalinf {a_2}$,
we have $\exec a {v_1}$ by theorem~\ref{eval-converges}. By induction
hypothesis, $\compute {n-1} {a_2} = \bot$. By
lemma~\ref{exec-either}, either $\compute {n-1} {a_1} = \bot$ or
$\compute {n-1} {a_1} = v_1$.  In both cases, $\compute n a = \bot$.
The third and last case ($\eval {a_1} {\lambda x.b}$ and
$\eval {a_2} {v_2}$ and $\evalinf {b[x \becomes v_2]}$) is similar.

The ``if'' part is proved by coinduction
and case analysis over $a$.  The cases $a = x$, $a = c$ and $a =
\lambda x.b$ trivially contradict the hypothesis $\exec a \bot$.
Therefore, it must be the case that $a = a_1~a_2$.  
Let $r_1$ and $r_2$ be the denotations of $a_1$ and $a_2$.  (They exist by
lemma~\ref{evaluates-total}.)  We argue by case over $r_1$ and $r_2$,
exploiting the definition of $\computename$ for sufficiently large
values of $n$.  There are only three cases that do not contradict the
hypothesis $\exec a \bot$: 
(1) $r_1 = \bot$;
(2) $r_1$ is a value $v_1$ and $r_2 = \bot$;
(3) $r_1$ is a value $\lambda x.b$ and $r_2$ is a value $v_2$
and $\exec {b[x \becomes v_2]} \bot$.
We conclude $\evalinf a$ by applying the appropriate inference rule 
for each case, the coinduction hypothesis for the
$\evalinfname$ premise, and theorem~\ref{eval-converges} for the
$\evalname$ premises.
\end{constrproof}

\section{Extension to trace semantics} \label{s-traces}

Besides expressing both terminating and diverging executions,
small-step semantics have another advantage over big-step semantics:
reduction sequences contain all intermediate reducts of the source
term in addition to its final value, therefore providing a complete
trace of the execution.  Such execution traces are useful both for
static analysis (by abstract interpretation of collecting semantics)
and to state and prove stronger semantic preservation properties for
program transformations.  In particular, when the input language is
imperative and features observable actions such as input/output,
traces of observable events are crucial to state and prove
observational equivalence results.

In this section, following the second
author's work \cite{Grall-phd}, we show how to extend the big-step
semantics of section~\ref{s-big-step} so that they produce
not only the outcome of an evaluation (final
value or divergence), but also a (possibly infinite) execution trace.

\subsection{Traces}

The traces we consider are finite or infinite sequences of terms
representing the intermediate reducts of the source program.
\begin{syntaxleft}
\syntaxclass{Finite traces:} &
t & ::= & \epsilon \alt a.t & \explanation{(inductive interpretation)}
\\
\syntaxclass{Infinite traces:} &
T & ::= & a.T & \explanation{(coinductive interpretation)}
\end{syntaxleft}%
By abuse of notation, we write $t.t'$ and $t.T$ for the concatenation
of a finite trace $t$ and a finite or infinite trace.  Concatenation
is associative and $\epsilon$ is a neutral element for concatenation.

If $t = a_1.a_2\ldots a_n$ is a finite trace, we define the left
application $t~b$ of this trace to a term $b$ and the right
application $v~t$ of a value $v$ to this trace as follows:
\begin{eqnarray*}
t~b & = & (a_1~b).(a_2~b)\ldots(a_n~b) \\
v~t & = & (v~a_1).(v~a_2)\ldots(v~a_n)
\end{eqnarray*}
We similarly define  the applications $T~b$ and $v~T$ where $T$ is an
infinite trace.

We define bisimilarity between infinite traces, written $T_1 \bisim T_2$,
by the following coinductive rule:
$$\iruledoublebasic{
         T_1 \bisim T_2
}{
         a.T_1 \bisim a.T_2
}$$
Concatenation and application of traces are compatible with
bisimilarity.

In set theory, bisimilarity is equivalent to equality.  In Coq's
constructive logic, bisimilarity is coarser than equality: there
exists infinite traces that are bisimilar but cannot be proved equal
\cite[chap. 13]{Bertot-Casteran-Coqart}.  Some of the following
results require the use of bisimilarity instead of equality in
definitions and statements, in order to be provable in Coq.  

\subsection{Small-step semantics with traces}

While our objective is to instrument big-step semantics to produce
execution traces, we start by doing this for the small-step semantics,
which is easier and helps us define precisely the traces we expect for
an execution.  For a finite reduction sequence
$a_1 \redonename a_2 \redonename \cdots \redonename a_{n-1} \redonename a_n$,
the expected (finite) trace is $t = a_1.a_2\ldots a_{n-1}$, that is, the
initial term and its intermediate reducts but not the final term.
Equivalently, the trace comprises the source terms for all reduction
steps performed in the sequence.  This is formalized by the following
rules for the predicate $\tred a t {a'}$ (read: ``$a$ reduces in zero,
one or several steps to $a'$ with trace $t$'').
\begin{ruleset}
\srule{\tred a \epsilon a }
\irule{
       \redone a {a'} \and \tred {a'} t {b}
}{
       \tred a {a.t} {b}
}
\end{ruleset}
For an infinite reduction sequence $a_1 \redonename \ldots \redonename
a_n \redonename \ldots$, the expected (infinite) trace is
$T = a_1 \ldots a_n \ldots$  This is captured by the following
coinductive rule defining the predicate $\tredinf a T$ (read: ``$a$
reduces infinitely with trace $T$'').
$$
\iruledoublebasic{
       \redone a b \and \tredinf b T
}{
       \tredinf a {a.T}
}
$$
It is intuitively clear that the small-step semantics with traces is
a refinement of that without traces.  We now formalize this intuition,
which is not obvious to prove constructively in the case of infinite
reductions.

\begin{lemma} \label{tred-red}
$\red a b$ if and only if $\exists t,~ \tred a t b$.
\end{lemma}
\begin{constrproof}
Straightforward by induction over the reduction sequences
$\red a b$ and $\tred a t b$.
\end{constrproof}

\begin{lemma} \label{tredinf-redinf}
$\redinf a$ if and only if $\exists T, ~ \tredinf a T$.
\end{lemma}
\begin{constrproof}
The ``if'' part is an easy proof by coinduction.  The ``only if'' part
is more involved: since the conclusion $\exists T, ~ \tredinf a T$ is
not a coinductively-defined predicate, we cannot reason directly by
coinduction.  Instead, we must construct explicitly a suitable
infinite trace $T$.  To this end, we first define a reduction function
$\cal R$ from terms to optional terms that is equivalent to the one-step
reduction predicate, that is
$$
{\cal R}(a) = \cases{ {\tt Some}(b), & if $\redone a b$; \cr
                      {\tt None},    & if $\notred a$. \cr}
$$
This function is total (by induction over $a$), therefore
proving that one-step reduction is decidable.  Next, to
every term $a$ we associate an infinite trace ${\cal T}(a)$ of all the
successive reducts of $a$.  This trace is defined, by guarded
corecursion, as
$$ 
{\cal T}(a) = \cases{
       a.{\cal T}(b), & if ${\cal R}(a) = {\tt Some}(b)$; \cr
       a.{\cal T}(a), & if ${\cal R}(a) = {\tt None}$. \cr
}$$
We then show that $\redinf a$ implies $\tredinf a {{\cal T}(a)}$.
This follows by coinduction from the fact that ${\cal T}(a) = a.{\cal T}(b)$
whenever $\redone a b$.
\end{constrproof}

As a corollary, we obtain the following analogue of
lemma~\ref{red-or-redinf}.

\begin{lemma} \label{tred-or-tredinf}
For all terms $a$, either there exist a term $b$ and a trace $t$ such that 
$\tred a t b$ and $\notred b$, or there exists an infinite trace $T$ such that
$\tredinf a T$.
\end{lemma}
\begin{classicproof}
Follows from lemmas \ref{red-or-redinf}, \ref{tred-red}
and \ref{tredinf-redinf}. 
\end{classicproof}

Additionally, the trace-based reduction relations are deterministic up
to bisimilarity between infinite traces.  This is an immediate
consequence of the determinism of one-step reductions
(lemma~\ref{red1-determ}).

\begin{lemma} \label{tred-determ}
If $\tred a {t_1} {v_1}$ and $\tred a {t_2} {v_2}$, then $t_1 =
t_2$ and $v_1 = v_2$.
\end{lemma}

\begin{lemma} \label{tredinf-determ}
If $\tredinf a {T_1}$ and $\tredinf a {T_2}$, then $T_1 \bisim T_2$.
\end{lemma}

Note that the stronger conclusion $T_1 = T_2$ is not provable in Coq.
Another consequence of the determinism of one-step reductions is the
following obvious decomposition property for infinite reductions.

\begin{lemma} \label{tredinf-decompose}
If $\tredinf a T$ and $\tred a t b$, there exists $T'$ such that
$\tredinf b {T'}$ and $T = t.T'$.
\end{lemma}

\subsection{Big-step semantics with traces}

We now add traces to the big-step definitions of evaluation and
divergence.  The corresponding predicates are $\teval a t v$
(``$a$ evaluates to $v$ with finite trace $t$'') and
$\tevalinf a T$ (``$a$ diverges with infinite trace $T$'').

\begin{ruleset}
\srulenumber{$\evalname$-const}{  \teval c \epsilon c   }
\srulenumber{$\evalname$-fun}{    \teval {\lambda x.a} \epsilon {\lambda x.a}  }%
\irulenumber{$\evalname$-app}{
    \teval {a_1} {t_1} {\lambda x.b} \and
    \teval {a_2} {t_2} {v_2} \and
    \teval {b[x \becomes v_2]} {t_3} v \\
    t = (t_1~a_2).((\lambda x.b)~t_2).((\lambda x.b)~v_2).t_3
}{
    \teval {a_1~a_2} {t} v
}
\\
\irulenumberdouble{$\evalinfname$-app-l}{
         \tevalinf {a_1} {T_1} \and T \bisim T_1~a_2
}{
         \tevalinf {a_1~a_2} {T}
}
\irulenumberdouble{$\evalinfname$-app-r}{
         \teval {a_1} {t_1} v \and \tevalinf {a_2} {T_2} \and
         T \bisim (t_1~a_2).(v~T_2)
}{
         \tevalinf {a_1~a_2} {T}
}
\irulenumberdouble{$\evalinfname$-app-f}{
         \teval {a_1} {t_1} {\lambda x.b} \and \teval {a_2} {t_2} {v_2} \and 
         \tevalinf {b[x \becomes v_2]} {T_3} \\
         T \bisim (t_1~a_2).((\lambda x.b)~t_2).((\lambda x.b)~v_2).T_3
}{
         \tevalinf {a_1~a_2} {T}
}
\end{ruleset}

The construction of the trace in the rules for applications 
is justified as follows.  Assume, for instance,
$\teval {a_1} {t_1} {\lambda x.b}$ and
$\teval {a_2} {t_2} {v_2}$.  The application $a_1~a_2$ performs one
$\beta$-reduction $\redone {(\lambda x.b)~v_2} {b[x \becomes v_2]}$
in addition to those coming from the evaluations of the premises of
the rule.  The source term for this reduction, $(\lambda x.b)~v_2$, is
therefore added to the trace.  It is preceded by $t_1~a_2$ (the trace
for $a_1$ put into a left application context $[\,]~a_2$) and by
$(\lambda x.b)~t_2$ (the trace for $a_2$ put into a right application
context $(\lambda x.b)~[\,]$).  The source of the $\beta$-reduction is
then followed by the trace corresponding to the evaluation of the
function body $b[x \becomes v_2]$.

Another point to note is the use of bisimilarity $T \bisim \ldots$
instead of equality $T = \ldots$ in the coinductive rules defining
$\evalinfname$.  This allows traces to be replaced by bisimilar traces
at every inference step, therefore enabling us to prove more
statements about $\evalinfname$ within the limits of Coq's coinductive proofs.
(For instance, the proof of theorem~\ref{tredinf-tevalinf} no longer
goes through if $\evalinfname$ is defined with equalities between traces
instead of bisimilarities.) This subtle point is moot in set theory,
where bisimilarity is equivalent to equality.

\begin{lemma} \label{tevalinf-omega}
$\tevalinf \omega T$ holds where $T$ is the infinite trace
$\omega.\omega.\omega\ldots$
\end{lemma}

\begin{proof} By coinduction, using rule ($\evalinfname$-app-f). \end{proof}

\subsection{Equivalence between the trace semantics}

We now show the equivalence between the big-step and small-step
semantics with traces, extending the results of section~\ref{s-small-step}.

\begin{theorem} \label{teval-tred-equiv}
$\teval a t v$ if and only if $\tred a t v$ and $\isvalue v$.
\end{theorem}

\begin{constrproof}  The ``only if'' part is an easy induction on
the derivation of $\teval a t v$.
For the ``if'' part, we first show the following two
lemmas: (1) $\teval v \epsilon v$ if $\isvalue v$, and 
(2) $\teval a {a.t} v$ if $\redone a b$ and $\teval b t v$.
The result follows by induction on the
derivation of $\tred a t v$.
\end{constrproof}

\begin{theorem} \label{tevalinf-tredinf}
$\tevalinf a T$ implies $\tredinf a T$.
\end{theorem}

\begin{constrproof} We first show by induction on $a$ that
$\tevalinf a T$ implies the existence of $b$ and $T'$ such that
$\redone a b$ and $\tevalinf b T'$ and $T \bisim a.T'$.  We then
define the following variant $\redinfbisimname$ of the infinite
reduction predicate, by the coinductive inference rule
$$\iruledoublebasic{
        \redone a b \and \tredinfbisim b T' \and T \bisim a.T'
}{
        \tredinfbisim a T
}$$
This variant enables us to replace the infinite trace $T$ by a
bisimilar one at every proof step, while remaining within the subset
of proofs that Coq accepts as productively coinductive.
We can therefore show that $\tevalinf a T$ implies $\tredinfbisim a T$
by coinduction, using the decomposition property stated earlier.
We conclude by proving that $\tredinfbisim a T$
implies $\tredinf a T$, again by coinduction.
\end{constrproof}

As a corollary of theorem~\ref{tevalinf-tredinf}, the big-step divergence
relation $\evalinfname$ is deterministic up to bisimilarity of the
traces.  It is interesting to note that we could not find a more
direct Coq proof of this fact.

\begin{lemma} \label{tevalinf-determ}
If $\tevalinf a {T_1}$ and $\tevalinf a {T_2}$, then $T_1 \bisim T_2$.
\end{lemma}

\begin{proof} Follows from lemma \ref{tredinf-determ} and theorem
\ref{tevalinf-tredinf}.
\end{proof}

The converse of theorem~\ref{tevalinf-tredinf} relies on the following
inversion lemma for infinite reduction sequences starting with an
application.

\begin{lemma} \label{tredinf-inv}
Assume $\tredinf {a~b} T$.
\begin{enumerate}
\item If $\tredinf a T'$, then $T \bisim T'~b$.
\item If $\isvalue a$ and $\tredinf b T'$, then $T \bisim a~T'$.
\item If $\tred a t {a'}$, then
there exists $T'$ such that $\tredinf {a'~b} {T'}$ and $T = (t~b).T'$.
\item If $\isvalue a$ and $\tred b t {b'}$, then
there exists $T'$ such that $\tredinf {a~b'} {T'}$ and $T = (a~t).T'$.
\end{enumerate}
\end{lemma}

\begin{constrproof} For (1) and (2), we show by coinduction that
$\tredinf {a~b} {T'~b}$ and $\tredinf {a~b} {a~T'}$, respectively,
then conclude by lemma~\ref{tredinf-determ}.

Property (3) follows from the decomposition
lemma~\ref{tredinf-decompose} and the fact that
$\tred {a~b} {t~b} {a'~b}$ whenever $\tred a t {a'}$.
Similarly, property (4) follows from the decomposition
lemma~\ref{tredinf-decompose} and the fact that
$\tred {a~b} {a~t} {a~b'}$ if $\isvalue a$ and $\tred b t {b'}$.
\end{constrproof}

\begin{theorem} \label{tredinf-tevalinf}
$\tredinf a T$ implies $\tevalinf a T$.
\end{theorem}

\begin{classicproof} The proof proceeds by coinduction and case
analysis over $a$.  It must be the case that $a = a_1~a_2$, otherwise
$a$ cannot reduce infinitely.  Using lemma~\ref{tred-or-tredinf},
we distinguish three cases:
\begin{enumerate}
\item $\tredinf {a_1} {T_1}$.
This implies $\tevalinf {a_1} {T_1}$ by coinduction hypothesis.
Moreover, we have $T \bisim T_1~a_2$ by case (1) of lemma~\ref{tredinf-inv},
which implies the expected result by rule~($\evalinfname$-app-l).

\item $\tred {a_1} {t_1} {v}$ and $\notred {v}$ and $\tredinf {a_2} {T_2}$.
By case (3) of lemma~\ref{tredinf-inv},
we have $\tredinf {v~a_2} {T'}$ for some $T'$ such that $T = (t_1~a_2).T'$.
This implies that $\isvalue v$.
Moreover, $T' \bisim v~T_2$ by case (2) of lemma~\ref{tredinf-inv}.
Theorem~\ref{teval-tred-equiv} gives $\teval {a_1} {t} {v}$
and the coinduction hypothesis gives $\evalinf {a_2} {T_2}$.
The result follows from rule~($\evalinfname$-app-r).

\item $\tred {a_1} {t_1} {v_1}$ and $\notred {v_1}$ and
$\tred {a_2} {t_2} {v_2}$ and $\notred {v_2}$.
Using cases (3) and (4) of lemma~\ref{tredinf-inv}, 
it follows that $v_1 = \lambda x.b$ for some $x$, $b$,
that $\isvalue {v_2}$, and that 
$\tredinf {(\lambda x.b)~v_2} {T'}$ for some $T'$
such that 
$T = (t_1~a_2).((\lambda x.b)~t_2).T'$.
By inversion, we deduce $\tredinf {b[x \becomes v_2]} {T_3}$
for some $T_3$ such that $T' \bisim ((\lambda x.b)~v_2). T_3$.
The result follows by rule~($\evalinfname$-app-f), the coinduction
hypothesis, and theorem~\ref{teval-tred-equiv}.
\end{enumerate}
\end{classicproof}

\section{Coevaluation} \label{s-coeval}

\subsection{Definition and properties} \label{s-coeval-1}

So far, we have described terminating and non-terminating evaluations
using two separate sets of inference rules, one interpreted
inductively and the other coinductively.  An attempt to describe both
kinds of evaluations at the same time, in a more concise way,
is to interpret coinductively the standard evaluation rules for
terminating evaluations.  This defines the relation $\coeval a b$ (read:
``$a$ coevaluates to $b$'').

\begin{ruleset}
\srulenumber{$\coevalname$-const}{  \coeval c c   }
\srulenumber{$\coevalname$-fun}{    \coeval {\lambda x.a} {\lambda x.a}  }
\irulenumberdouble{$\coevalname$-app}{
    \coeval {a_1} {\lambda x.b} \and
    \coeval {a_2} {v_2} \and
    \coeval {b[x \becomes v_2]} v
}{
    \coeval {a_1~a_2} v
}
\end{ruleset}

It is clear from the definition of $\coevalname$ that coevaluation
includes all terminating evaluations, plus some diverging ones.

\begin{lemma} \label{eval-coeval}
If $\eval a v$, then $\coeval a v$. 
\end{lemma}
\begin{constrproof} By induction on the derivation of $\eval a v$.
\end{constrproof}

\begin{lemma} \label{coeval-omega}
$\coeval \omega v$ for all terms $v$. \end{lemma}
\begin{constrproof} By coinduction, using rule ($\coevalname$-app) with
the coinduction hypothesis as third premise. \end{constrproof}

Naively, we could expect that $\coevalname$~is equivalent to the union
of the $\evalname$~and~$\evalinfname$ relations.  This equivalence
holds in one direction only, from coevaluation to evaluation.

\begin{lemma}  \label{coeval-eval-or-evalinf}
If $\coeval a v$, then either $\eval a v$ or $\evalinf a$.
\end{lemma}
\begin{classicproof} We show that $\coeval a v$ and $\lnot(\eval a v)$
implies $\evalinf a$.  The result then follows by
excluded middle on $\eval a v$.  The auxiliary property is proved
by coinduction and case analysis on $a$.
The cases for variables, constants and abstractions trivially
contradict one of the hypotheses.
If $a = a_1~a_2$, an inversion on the hypothesis $\coeval a v$
shows that $\coeval {a_1} {\lambda x.b}$ and $\coeval {a_2} {v_2}$
and $\coeval {b[x \becomes v_2]} v$.  Using excluded middle, it must be
that at least one of these three terms does not evaluate, otherwise,
$\eval a v$ would hold.  The result follows by applying the rule for
$\evalinfname$~that matches the term that does not evaluate, and using the
coinduction hypothesis.
\end{classicproof}

However, the reverse implication from evaluation to coevaluation 
does not hold: there exists terms
that diverge but do not coevaluate.  Consider for instance $a =
\omega~(0~0)$.  It is true that $\evalinf a$, but there is no term $v$
such that $\coeval a v$, because the coevaluation of the argument
$0~0$ goes wrong (there is no $v$ such that $\coeval {0~0} v$).
Section~\ref{s-soundness-bigstep} shows another example of a diverging
term that does not coevaluate, this time involving no subterm that
goes wrong.

Another unusual feature of coevaluation is that it is not
deterministic. For instance, $\coeval \omega v$ for any term $v$.
However, $\coevalname$ is deterministic for terminating terms, in the
following sense:

\begin{lemma} If $\eval a v$ and $\coeval a {v'}$, then $v' = v$.
\end{lemma}

\begin{constrproof} By induction on the derivation of $\eval a v$
and inversion on $\coeval a {v'}$.
\end{constrproof}

Moreover, there exists diverging terms that coevaluate to only one value.
An example is $(\lambda x.0)~\omega$, which coevaluates to $0$ but not
to any other term.

\subsection{Connection with small-step semantics} \label{s-coeval-2}

Concerning the connections between coevaluation (big-step) and
coreduction (small-step) in the style of section~\ref{s-small-step},
the expected equivalence between $\coevalname$ and $\coredname$
holds in one direction only.

\begin{lemma} $\coeval a v$ implies $\cored a v$. \end{lemma}

\begin{constrproof} Using classical logic, this follows from lemmas
\ref{coeval-eval-or-evalinf} and equivalence theorems \ref{eval-red-equiv},
\ref{evalinf-redinf-equiv} and \ref{cored-red-or-redinf}.  However,
the result can be proved directly in constructive logic.
We first show that $\coeval a v \implies \isvalue a \lor \exists b,~
\redone a b \land \coeval b v$ by induction on $a$.  The result
follows by coinduction.
\end{constrproof}

The reverse implication obviously does not hold for terms $a$ that
diverge but do not coevaluate, such as the term $a = \omega~(0~0)$
mentioned previously: if $\evalinf a$, we have $\redinf a$ and
therefore $\cored a v$ for any $v$, but $\coeval a v$ does not hold.
Another counterexample to the reverse implication is
$a = (\lambda x.~0)~\omega$ and $v = 1$.  Since $\redinf a$, we have
$\cored a v$.  However, $\coeval a v$ does not hold since the only
term to which $a$ coevaluates is $0$.

\subsection{Coevaluation for CPS terms} \label{s-coeval-CPS}

Notwithstanding the negative results of
sections~\ref{s-coeval-1}~and~\ref{s-coeval-2}, there exists a class
of terms for which coevaluation correctly captures both terminating
and diverging evaluations: terms that are in continuation-passing
style (CPS).  A distinguishing feature of these terms is that function
arguments are always values.  
CPS terms are defined by the following grammar:
$$\begin{array}{l@{~~}c@{~~}l}
a \in \mbox{\sf Atoms} & ::= & x \alt c \alt \lambda x. b \\
b \in \mbox{\sf CPS-terms} & ::= & a \alt b~a
\end{array}$$
Less formally, CPS terms are built from atoms (variables,
constants and function abstractions) using multiple applications in
tail-call position.

It is well known that CPS terms are stable by substitution of atoms
for variables.

\begin{lemma} \label{isbody-subst}
If $a \in \mbox{\sf Atoms}$ and $b \in \mbox{\sf CPS-terms}$, then
$b[x \becomes a] \in \mbox{\sf CPS-terms}$.
\end{lemma}

Consequently, the value of a CPS term is an atom.

\begin{lemma} \label{eval-body}
If $b \in \mbox{\sf CPS-terms}$ and $\eval b v$, then
$v \in \mbox{\sf Atoms}$.  As a corollary, if 
 $b \in \mbox{\sf CPS-terms}$ and $\eval b {\lambda x.b'}$,
then $b' \in  \mbox{\sf CPS-terms}$.
\end{lemma}

\begin{constrproof} By induction on the derivation of $\eval b v$,
using lemma~\ref{isbody-subst} for the application case.
\end{constrproof}

The main result of this section is that a closed CPS term coevaluates
to a value if and only if it evaluates or it diverges.  The
restriction to closed terms is important since, for instance,
the CPS term $\omega~x$ diverges but its coevaluation goes wrong on
the free variable $x$.

The following lemma lists useful properties of CPS atoms.

\begin{lemma}  \label{atom-props}
Let $a \in \mbox{\sf Atoms}$.
\begin{enumerate}
\item $\eval a a$ if $a$ is closed.
\item It is not the case that $\evalinf a$.
\item If $\eval a v$, then $v = a$.
\end{enumerate}
\end{lemma}

The key technical lemma below shows that diverging, closed CPS terms
coevaluate to a well-chosen value.

\begin{lemma} \label{evalinf-coeval}
Define $\Omega = \lambda x.\omega$.  
If $b \in \mbox{\sf CPS-terms}$, $b$ is closed and $\evalinf b$,
then $\coeval b \Omega$.
\end{lemma}

\begin{constrproof} By coinduction.  The CPS term $b$ cannot be
an atom (this would contradict the divergence hypothesis), therefore
$b = b'~a$ with $b'$ a closed CPS term and $a$ a closed CPS atom.
Analysis on the last rule used in the derivation of $\evalinf b$
reveals three cases.  In the first case, $\evalinf {b'}$.  By
coinduction hypothesis, $\coeval {b'} {\Omega = \lambda x.\omega}$.  
By lemmas~\ref{atom-props}~and~\ref{eval-coeval}, $\coeval a a$.
Finally, $\omega[x \becomes a] = \omega$ coevaluates to $\Omega$
by lemma~\ref{coeval-omega}.  Applying rule~($\coevalname$-app),
it follows that $\coeval b \Omega$.

The second case, $\evalinf a$, is impossible by lemma~\ref{atom-props}.
This leaves the third case: $\eval {b'} {\lambda x.b''}$ and $\eval a v$
and $\evalinf {b''[x \becomes v]}$.
By lemma~\ref{eval-body}, $b''$ is a CPS term. By lemma~\ref{atom-props}, 
$v = a$ and therefore $v$ is a CPS atom.  It follows that
$b''[x \becomes v]$ is a CPS term (lemma~\ref{isbody-subst}).
Moreover, this term is closed because of the usual properties of free
variables w.r.t. evaluation and substitution.  Using
lemma~\ref{eval-coeval} and the coinduction hypothesis, we obtain
$\coeval {b'} {\lambda x.b''}$ and $\coeval a v$ and $\coeval {b''[x
\becomes v]} \Omega$, from which $\coeval b \Omega$ follows by
rule~($\coevalname$-app).
\end{constrproof}

The claimed equivalence result follows as a corollary.

\begin{theorem} Let $b$ be a closed CPS term.
We have $\exists v, \coeval b v$ if and only if
$\evalinf b$ or $\exists v, \eval b v$.
\end{theorem}

\begin{constrproof} Follows from lemmas
\ref{eval-coeval}, \ref{coeval-eval-or-evalinf} and
\ref{evalinf-coeval}.
\end{constrproof}

\section{Type soundness proofs} \label{s-type-soundness}

We now turn to using our coinductive evaluation and reduction
relations for proving the soundness of type systems.  To be more
specific, we will use the simply-typed $\lambda$-calculus with
recursive types as our type system.  We obtain recursive types
by interpreting the type algebra $\tau ::= {\tt int} \alt \tau_1 \fun
\tau_2$ coinductively, as in \cite{Gapeyev-Levin-Pierce}.
The typing rules are recalled below.
Type environments, written $E$, are finite maps from
variables to types.

\begin{ruleset}
\irule{   E(x) = \tau
}{
         E \vdash x : \tau
}
\srule{  E \vdash c : {\tt int}  }
\\
\irule{
         E + \{ x : \tau' \} \vdash a : \tau
}{
         E \vdash \lambda x.a : \tau' \fun \tau
}
\irule{
         E \vdash a_1 : \tau' \fun \tau \and E \vdash a_2 : \tau'
}{
         E \vdash a_1 ~ a_2 : \tau
}
\end{ruleset}
Enabling recursive types makes the type system non-normalizing
and makes it possible to write interesting programs.  In
particular, the call-by-value fixpoint operator
$ Y = \lambda f.~ (\lambda x.~f~(x~x))
                  ~(\lambda x.~f~(\lambda y.~(x~x)~y)) $
is well-typed, with types $((\tau \fun \tau') \fun \tau
\fun \tau') \fun \tau \fun \tau'$ for all types~$\tau$ and~$\tau'$.
(The self-applications $x~x$ are well-typed under the assumption $x : \sigma$,
where the recursive type $\sigma$ is defined by the equation
$\sigma = \sigma \fun \tau \fun \tau'$.)

\subsection{Type soundness proofs using small-step semantics}

Wright and Felleisen \cite{Felleisen-Wright} introduced a proof
technique for showing type soundness that relies on small-step semantics
and is standard nowadays.  The proof relies on the twin
properties of {\em type preservation} (also called {\em subject
reduction\/}) and {\em progress\/}:

\begin{lemma}[Preservation]
If $\redone a b$ and $\emptyset \vdash a : \tau$, then $\emptyset \vdash b : \tau$
\end{lemma}

\begin{lemma}[Progress]
If $\emptyset \vdash a : \tau$, then either $\isvalue a$ or there exists
 $b$ such that $\redone a b$.
\end{lemma}

The formal statement of type soundness in Felleisen and Wright's
approach is the following:

\begin{theorem}[Type soundness, 1]
If $\emptyset \vdash a : \tau$ and $\red a b$, then either $\isvalue b$
or $b$ reduces.
\end{theorem}

\begin{constrproof} We first show that $\emptyset \vdash b : \tau$ by
induction over $\red a b$, using the preservation lemma.  We then
conclude with the progress lemma.
\end{constrproof}

The authors that follow this approach then conclude that well-typed closed
terms either reduce to a value or reduce infinitely.  However, this
conclusion is generally neither expressed nor proved formally.  In our
approach, it is easy to do so:

\begin{theorem}[Type soundness, 2] \label{type-soundness-2}
If $\emptyset \vdash a : \tau$, then either $\redinf a$, or there exists
$v$ such that $\red a v$ and $\isvalue v$.
\end{theorem}

\begin{classicproof}  By lemma~\ref{red-or-redinf}, either $\redinf a$
or $\exists b,~\red a b \land \notred b$.  The result is obvious in
the first case.  In the second case, we note that $\emptyset \vdash b : \tau$
as a consequence of the preservation lemma, then use the progress
lemma to conclude that $\isvalue b$.
\end{classicproof}

An alternate, equivalent formulation of this theorem uses the
coreduction relation $\coredname$.

\begin{theorem}[Type soundness, 3]
If $\emptyset \vdash a : \tau$, then there exists $v$ such that
$\cored a v$ and $\isvalue v$.
\end{theorem}

\begin{constrproof}  Follows from theorem~\ref{type-soundness-2}
and lemma~\ref{cored-red-or-redinf}.
\end{constrproof}

An arguably nicer characterisation of ``programs that do not go
wrong'' is given by the relation $\safered a$ (read: ``$a$ reduces
safely''), defined coinductively by the following rules:
\begin{ruleset}
\iruledouble{
    \isvalue v
}{
    \safered v
}
\iruledouble{
    \redone a b \and \safered b
}{
    \safered a
}
\end{ruleset}
These rules are interpreted coinductively so that $\safered a$ holds
if $a$ reduces infinitely.  We can then state and show type soundness
without recourse to classical logic:

\begin{theorem}[Type soundness, 4]
If $\emptyset \vdash a : \tau$, then $\safered a$.
\end{theorem}

\begin{constrproof}  By coinduction.  Applying the progress lemma, either
$\isvalue a$ and we are done, or $\redone a b$ for some $b$.
In the latter case,
$\emptyset \vdash b : \tau$ by the preservation property, and the result
follows from the coinduction hypothesis.
\end{constrproof}

\subsection{Type soundness proofs using big-step semantics}
\label{s-soundness-bigstep}

The standard big-step semantics (defined by the $\evalname$ relation) is awkward
for proving type soundness because it does not distinguish between
terms that diverge and terms that go wrong: in both cases, there is no
value $v$ such that $\eval a v$.  Consequently, the obvious type
soundness statement ``if $\emptyset \vdash a : \tau$, there exists $v$
such that $\eval a v$'' is false for all type systems that do not
guarantee normalization.  The best result we can prove, then,
is the following big-step equivalent to the preservation lemma:

\begin{lemma}[Preservation, big-step style] \label{eval-preservation}
If $\eval a v$ and $\emptyset \vdash a : \tau$, then $\emptyset \vdash v : \tau$.
\end{lemma}
\begin{constrproof} Easy induction on the derivation of $\eval a v$, using
  the fact that typing is stable by substitution: if
$\{ x: \tau' \} \vdash a : \tau$ and $\emptyset \vdash b : \tau'$,
then $\emptyset \vdash a[x \becomes b] : \tau$.
\end{constrproof}

The standard approach for proving type soundness using big-step
semantics is to provide inductive inference rules to define
a predicate $\evalerr a$ characterizing terms that go wrong because of
a type error, 
and prove the statement ``if $\emptyset \vdash a :
\tau$, then it is not the case that $\evalerr a$'' \cite{Tofte-phd}.
This approach is not fully satisfactory for two reasons: (1) extra rules must
be provided to define $\evalerr a$, which increases the size of the
semantics; (2) there is a risk that the rules for $\evalerr a$ are
incomplete and miss some cases of ``going wrong'', in which case the
type soundness statement does not guarantee that well-typed terms
either evaluate to a value or diverge.

Let us revisit these trade-offs in the light of our characterizations
of divergence and coevaluation.  We can now formally state what it
means for a term to evaluate or to diverge.  This leads to the
following alternate statement of type soundness:

\begin{theorem}[Type soundness, 5] \label{type-soundness-5}
If $\emptyset \vdash a : \tau$, then either $\evalinf a$ or there exists
$v$ such that $\eval a v$.
\end{theorem}

By excluded middle, either $\exists v.~\eval a v$ or
$\forall v, ~ \lnot(\eval a v)$.  Theorem~\ref{type-soundness-5}
therefore follows from lemma~\ref{evalinf-progress} below,
which is a big-step analogue to the progress lemma.

\begin{lemma}[Progress, big-step style] \label{evalinf-progress}
If $\emptyset \vdash a : \tau$ and $\forall v, ~ \lnot(\eval a v)$,
then $\evalinf a$.
\end{lemma}

\begin{classicproof}  The proof is by coinduction and case analysis
over $a$.  The cases $a = x$, $a = c$ and $a = \lambda x.b$ lead to
contradictions: variables have no types in the empty environment;
constants and abstractions evaluate to themselves.  The interesting
case is therefore $a = a_1~a_2$.  By excluded middle, either $a_1$
evaluates to some value $v_1$, or not.  In the latter case, $\evalinf
a$ follows from rule ($\evalinfname$-app-l) and from $\evalinf {a_1}$,
which we obtain by coinduction hypothesis.  In the former case, $v_1$
has a function type $\tau' \fun \tau$ by lemma~\ref{eval-preservation},
and therefore $v_1 = \lambda x.b$ for some $x$
and $b$.  Moreover, $\{x : \tau' \} \vdash b : \tau$.  Using excluded
middle again, either $a_2$ evaluates to some value $v_2$, or not.  In
the latter case, $\evalinf a$ follows from rule ($\evalinfname$-app-r)
and the coinduction hypothesis.  In the former case,
$\emptyset \vdash v_2 : \tau'$.  Since typing is stable by substitution,
$\emptyset \vdash b[x \becomes v_2] : \tau$.  Using excluded middle for
the third time, it must be that $\forall v.~\lnot(\eval {b[x \becomes
v_2]} v)$, otherwise $a$ would evaluate to some value.  The result
$\evalinf a$ then follows from rule ($\evalinfname$-app-f) and the
coinduction hypothesis.
\end{classicproof}

The proof above is an original alternative to the standard approach of
showing $\lnot (\evalerr a)$ for all well-typed terms $a$.  From a
methodological standpoint, our proof addresses one of the shortcomings
of the standard approach, namely the risk of not putting in enough error
rules.  If we forget some divergence rules,
the proof of lemma~\ref{evalinf-progress} will, in all likelihood, not
go through.  Therefore, this novel approach to proving
type soundness using big-step semantics appears rather robust with
respect to mistakes in the specification of the semantics.

The other methodological shortcoming remains, however: just like the
``not goes wrong'' approach, our approach requires more evaluation
rules than just those for normal evaluations, namely the rules for
divergence.  This can easily double the size of the specification of a
dynamic semantics, which is a concern for realistic languages
where the normal evaluation rules number in dozens.  

The coevaluation relation $\coevalname$ is attractive for this
pragmatic reason, as it has the same number of rules as normal
evaluation.  Of course, we have seen that $\coeval a v$ is not
equivalent to $\eval a v \lor \evalinf a$, but the example we gave was
for a diverging term $a$ that is not typeable and where an early diverging
evaluation ``hides'' a later evaluation that goes wrong.  Since type
systems ensure that all subterms of a term do not go wrong, we could
hope that the following conjecture holds:

\begin{conjecture}[Type soundness, 6] \label{type-soundness-6}
If $\emptyset \vdash a : \tau$, there exists $v$ such that $\coeval a v$.
\end{conjecture}

We were able to prove this conjecture for some uninteresting but
nonetheless non-normalizing type systems, such as simply-typed
$\lambda$-calculus without recursive types, but with a predefined
constant of type ${\tt int} \fun {\tt int}$ that diverges when applied.
However, the conjecture is false for simply-typed
$\lambda$-calculus with recursive types, and probably 
for all type systems with a general fixpoint operator.  Andrzej
Filinski provided the following counterexample. Consider
$$ Y~F~0 \qquad \mbox{where} \qquad
   F = \lambda f.\lambda x.~(\lambda g.\lambda y.~g~ y)~(f~x) $$
or, in more readable ML notation
\begin{verbatim}
     let rec f x = (let g = f x in fun y -> g y) in f 0
\end{verbatim}
The term $Y~F~0$ is well-typed with type $\tau \fun \tau'$, yet it
fails to coevaluate: the only possible value $v$ 
such that $\coeval {Y~F~0} v$
would be an infinite term, $\lambda y.~ (\lambda y.~ (\lambda y.~\ldots~y)~~y)~y$.

\section{Compiler correctness proofs} \label{s-compiler-correctness}

We now return to the original motivation of this work: proving
that compilers preserve the semantics of source programs (including
diverging ones), using big-step semantics.  We demonstrate this
approach on the compilation of call-by-value $\lambda$-calculus down
to a simple abstract machine.

\subsection{Big-step semantics with environments and closures}

Our abstract machine uses closures and environments indexed by de
Bruijn indices.  It is therefore
convenient to reformulate the big-step evaluation predicates in these
terms.  Variables, written $x_n$, are now identified by their de Bruijn
indices~$n$. 
Values (which are no longer a subset of terms) and
environments are defined as:
\begin{syntaxleft}
\syntaxclass{Values:}
& v & ::=  & c & \explanation{integer values} \\
&   & \alt & (\lambda a)[e] & \explanation{function closures} \\
\syntaxclass{Environments:}
& e & ::=  & \epsilon \alt v.e & \explanation{sequences of values}
\end{syntaxleft}
As in section~\ref{s-big-step}, we define three
evaluation relations by the inference rules given below.
$$\begin{array}{ll}
e \vdash \eval a v & \mbox{~~~~finite evaluations (inductive)} \\
e \vdash \evalinf a & \mbox{~~~~infinite evaluations (coinductive)} \\
e \vdash \coeval a v & \mbox{~~~~coevaluations (coinductive)}
\end{array}$$

\begin{ruleset}
\irule{
    e = v_1 \ldots v_n \ldots
}{
    e \vdash \eval {x_n} {v_n}
}
\srule{
    e \vdash \eval c c
}
\srule{
    e \vdash \eval {\lambda a} {(\lambda a)[e]}
}
\irule{
    e \vdash \eval {a_1} {(\lambda b)[e']} \and
    e \vdash \eval {a_2} {v_2} \and
    v_2.e' \vdash \eval b v
}{
    e \vdash \eval {a_1 ~ a_2} v
}
\\
\iruledouble{
         e \vdash \evalinf {a_1}
}{
         e \vdash \evalinf {a_1~a_2}
}
\iruledouble{
         e \vdash \eval {a_1} v \and e \vdash \evalinf {a_2}
}{
         e \vdash \evalinf {a_1~a_2}
}
\iruledouble{
         e \vdash \eval {a_1} {(\lambda b)[e']} \and e \vdash \eval {a_2} v \and 
         v.e' \vdash \evalinf b
}{
         e \vdash \evalinf {a_1~a_2}
}
\\
\iruledouble{
    e = v_1 \ldots v_n \ldots
}{
    e \vdash \coeval {x_n} {v_n}
}
\srule{
    e \vdash \coeval c c
}
\srule{
    e \vdash \coeval {\lambda a} {(\lambda a)[e]}
}
\iruledouble{
    e \vdash \coeval {a_1} {(\lambda b)[e']} \and
    e \vdash \coeval {a_2} {v_2} \and
    v_2. e' \vdash \coeval b v
}{
    e \vdash \coeval {a_1 ~ a_2} v
}
\end{ruleset}

We will not formally study these relations, but note that they enjoy
the same properties as the environment-less relations studied in
section~\ref{s-big-step}.

\subsection{The abstract machine and its compilation scheme}

The abstract machine we use as target of compilation follows the
call-by-value strategy and the ``eval-apply'' model \cite{Peyton-Jones-Marlow-06}.  It is close in
spirit to the SECD, CAM, FAM and CEK machines
\cite{Landin-SECD,Cousineau-Curien-Mauny-CAM,%
Cardelli-FAM,Felleisen-Friedman-86}.
The machine state has three components: a code sequence, a stack and
an environment.  The syntax for these components is as follows.
\begin{syntaxleft}
\syntaxclass{Instructions:}
& I & ::=  & {\tt Var}(n)  & \explanation{push the value of variable number $n$} \\
&   & \alt & {\tt Const}(c) & \explanation{push the constant $c$} \\
&   & \alt & {\tt Clos}(C)  & \explanation{push a closure for code $C$} \\
&   & \alt & {\tt App}     & \explanation{perform a function application} \\
&   & \alt & {\tt Ret}  & \explanation{return to calling function} \\
\syntaxclass{Code:}
& C & ::=  & \epsilon \alt I, C  & \explanation{instruction sequences} \\
\syntaxclass{Values:}
& V & ::=  & c  & \explanation{constant values} \\
&   & \alt & C[E] & \explanation{code closures} \\
\syntaxclass{Environments:}
& E & ::=  & \epsilon \alt V.E \\
\syntaxclass{Stacks:}
& S & ::=  & \epsilon  & \explanation{empty stack} \\
&   & \alt & V.S & \explanation{pushing a value} \\
&   & \alt & (C, E). S & \explanation{pushing a return frame} 
\end{syntaxleft}%
The behaviour of the abstract machine is defined
as a transition relation
$\redone {C; S; E} {C'; S'; E'}$ that relates the machine states
($C;S;E$) and ($C';S';E'$) respectively before and after the
execution of the first instruction of the code~$C$.
The transitions are as follows.

$$
\def\arraystretch{1.5}
\begin{array}{l|l|l||l|l|l@{}l}
\multicolumn{3}{c||}{\mbox{State before transition}} &
\multicolumn{3}{|c}{\mbox{State after transition}} \\
\mbox{~Code} & \mbox{~Stack} & \mbox{~Env.~} &
\mbox{~Code~} & \mbox{~Stack} & \mbox{~Env.~} \\
\cline{1-6}
{\tt Var}(n), C & ~S & ~E &
~C & ~V_n.S & ~E
& \mbox{if $E = V_1\ldots V_n \ldots$}
\\
{\tt Const}(c), C ~ & ~S & ~E &
~C & ~c.S & ~E
\\
{\tt Clos}(C'), C & ~S & ~E &
~C & ~C'[E].S & ~E
\\
{\tt App}, C & ~V.C'[E'].S & ~E &
~C' & ~(C,E).S & ~V.E'
\\
{\tt Ret}, C & ~V.(C',E').S & ~E &
~C' & ~V.S & ~E'
\end{array}$$

As in section~\ref{s-small-step}, we consider the following closures of the
one-step transition relation:
$$\begin{array}{rcl@{~~~~}l}
C;S;E & \redname & C';S';E' &
  \mbox{zero, one or several transitions (inductive)} \\
C;S;E & \redplusname & C';S';E' &
  \mbox{one or several transitions (inductive)} \\
C;S;E & \redinfname & &
  \mbox{infinitely many transitions (coinductive)} \\
C;S;E & \coredname & C';S';E' &
  \mbox{zero, one, several or infinitely many} \\
& & & \mbox{transitions (coinductive)}
\end{array}$$

The compilation scheme from terms to code is straightforward:
\begin{eqnarray*}
\compile {x_n} & = & {\tt Var}(n) \\
\compile {c}   & = & {\tt Const}(c) \\
\compile {\lambda a} & = & {\tt Clos}(\compile{a}, {\tt Ret}) \\
\compile {a_1~a_2} & = & \compile{a_1}, \compile{a_2}, {\tt App}
\end{eqnarray*}
The intended effect for the code $\compile a$ is to evaluate the term
$a$ and push its value at the top of the machine stack, leaving the
rest of the stack and the environment unchanged.

\subsection{Proofs of semantic preservation} \label{s-semantic-preservation}

We expect the compilation to abstract machine code to preserve the
behaviour of the source term, in the following general sense.
Consider a closed term $a$ and start the abstract machine in the
initial state corresponding to $a$.  If $a$ diverges, the machine
should perform infinitely many transitions.  If $a$ evaluates to
the value~$v$, the machine should reach a final state corresponding to
$v$ in a finite number of transitions.  Here, the initial state
corresponding to $a$ is $\compile a; \epsilon; \epsilon$.  
The final state corresponding to the result value $v$ is
$\epsilon; \compile{v}.\epsilon; \epsilon$, that is, the code has been
entirely consumed and the machine value $\compile{v}$ corresponding to
the source-level value $v$ is left on top of the stack.  The
correspondence between source-level values and machine values, as well
as between source-level environments and machine environments,
is defined by:
$$\compile{c} = c \qquad
\compile{(\lambda a)[e]} = (\compile{a}, {\tt Ret})[\compile e] \qquad
\compile{v_1 \ldots v_n} = \compile{v_1} \ldots \compile{v_n}$$

Semantic preservation is easy to show for terminating terms $a$ using
the big-step semantics.  We just need to strengthen the statement of
preservation so that it lends itself to induction over the derivation
of $e \vdash \eval a v$.

\begin{theorem} \label{compile-eval}
If $e \vdash \eval a v$, then
$ \redplus {(\compile{a}, C) ;~ S ;~ \compile{e}}
            {C ;~ \compile{v}.S ;~ \compile{e}} $
for all codes~$C$ and stacks~$S$.
\end{theorem}

\begin{constrproof} By induction on the derivation of $e \vdash \eval a v$.
The base cases where $a$ is a variable, a constant or an abstraction
are straightforward.  The inductive case is $a = a_1~a_2$ with
$e \vdash \eval {a_1} {(\lambda b)[e']}$ and
$e \vdash \eval {a_2} {v_2}$ and
$v_2.e' \vdash \eval b v$.  We build the following sequence of machine
transitions:
$$
\def\comment#1{ & \mbox{\hspace*{2cm} (#1)}}
\begin{array}{cl}
&
   (\compile{a_1}, \compile{a_2}, {\tt App}, C) ;~ S ;~ \compile{e} \\
\comment{induction hypothesis applied to the evaluation of $a_1$} \\
\redplusname &
   (\compile{a_2}, {\tt App}, C) ;~ \compile{(\lambda b)[e']}.S ;~ \compile{e} \\
\comment{induction hypothesis applied to the evaluation of $a_2$} \\
\redplusname &
   ({\tt App}, C) ;~ \compile{v_2}.\compile{(\lambda b)[e']}.S ;~ \compile{e} \\
\comment{{\tt App} transition, since 
      $\compile{(\lambda b)[e']} = (\compile{b}, {\tt Ret})[\compile{e'}]$} \\
\redonename &
   (\compile{b}, {\tt Ret}) ;~ (C, \compile{e}).S ;~ \compile{v_2}.\compile{e'} \\
\comment{induction hypothesis applied to the evaluation of $b$} \\
\redplusname &
   {\tt Ret} ;~ \compile{v}.(C, \compile{e}).S ;~ \compile{v_2}.\compile{e'} \\
\comment{{\tt Ret} transition} \\
\redonename &
   C ;~ \compile{v}.S ;~ \compile{e}
\end{array}$$
The result follows by transitivity of $\redplusname$.
\end{constrproof}

It is impossible, however, to prove semantic preservation for
diverging terms using only the standard big-step semantics, since
it does not describe divergence.  This led
several authors to prove semantic preservation for compilation to
abstract machines using small-step semantics with explicit
substitutions \cite{Rittri-88,Hardin-Maranget-Pagano}.  To this end,
they prove a simulation result between machine transitions and
source-level reductions: every machine transition corresponds to zero
or one source-level reductions.  To make the correspondence precise,
they need to define a {\em decompilation} relation that maps
intermediate machine states back to source-level terms.  However,
decompilation relations are difficult to define, especially for
optimizing compilation schemes; see \cite[section 4.3]{Gregoire-phd} for an
example.

The coinductive big-step semantics studied in this article provide a
simpler way to prove semantic preservation for non-terminating terms.
Namely, the following two theorems hold, showing that compilation
preserves divergence and coevaluation as characterized by the
$\evalinfname$ and $\coevalname$ predicates.

\begin{theorem} \label{compile-evalinf}
If $e \vdash \evalinf a$, then
$ \redinf {(\compile{a}, C) ;~ S ;~ \compile{e}} $
for all codes~$C$ and stacks~$S$.
\end{theorem}

\begin{theorem} \label{compile-coeval}
If $e \vdash \coeval a v$, then
$ \cored   {(\compile{a}, C) ;~ S ;~ \compile{e}}
            {C ;~ \compile{v}.S ;~ \compile{e}} $
for all codes~$C$ and stacks~$S$.
\end{theorem}

Both theorems cannot be proved directly by coinduction and
case analysis over~$a$.  The problem is in the
application case $a = a_1 ~ a_2$, where the code component of the
initial machine state is of the form $\compile{a_1}, \compile{a_2},
{\tt App}, C$.  It is not possible to invoke the coinduction hypothesis to
reason over the execution of $\compile{a_1}$, because this use of the
coinduction hypothesis is not guarded by an inference rule for the
$\redinfname$ relation, or in other terms because no machine
instruction is executed before invoking the hypothesis.  In the
approach to coinduction based on systems of equations presented in
section~\ref{s-proof-approach}, the problem manifests itself as a
non-guarded equation
$x_j = x_{j'}$ when $j$ is the judgment
$\redinf {(\compile{a_1~a_2}, C) ;~ S ;~ \compile{e}}$
associated with the state
$e \vdash \coeval {a_1~a_2} v$, $C$ and $S$,
while $j'$ is the equivalent judgment
$\redinf {(\compile{a_1}, (\compile{a_2}, {\tt App}, C)) ;~ S ;~ \compile{e}}$
associated with the state
$e \vdash \coeval {a_1} v$, $(\compile{a_2}, {\tt App}, C)$ and $S$.

There are two ways to address this issue.  The first is to modify the
compilation scheme for applications, in order to insert a ``no
operation'' instruction in front of the generated sequence:
$ \compile{a_1~a_2} = {\tt Nop}, \compile{a_1}, \compile{a_2} $.
The {\tt Nop} operation has the obvious machine transition
$ \redone {({\tt Nop}, C) ;~ S ;~ E} {C ;~ S ;~ E} $.
With this modification, the coinductive proof for
lemma~\ref{compile-evalinf} performs a {\tt Nop} transition before
invoking the coinduction hypothesis to deal with the evaluation of
$\compile{a_1}$.  This makes the coinductive proof properly guarded.

Of course, it is inelegant to pepper the generated code with {\tt Nop}
instructions just to make one proof go through.  We therefore use an
alternate approach where the compilation scheme for applications is
unchanged, but we exploit the fact that the
number of such recursive calls that do not perform a machine
transition is necessarily finite, because our term algebra is finite.
More precisely, this number is the left application height
$\leftappheight a$ of the term $a$ being compiled, where
$\leftappheight a$ is defined by
$$
\leftappheight {a_1~a_2} = \leftappheight{a_1} + 1
\qquad\qquad
\leftappheight x = \leftappheight c = \leftappheight {\lambda a} = 0
$$

To prove theorem~\ref{compile-evalinf}, we follow the approach described
by Bertot \cite{Bertot-05} in his coinductive presentation and proof of
Eratosthenes' sieve algorithm.  We first define the coinductive
relation $\redinfNname{n}$ where $n$ is a nonnegative integer:
\begin{ruleset}
\irulenumberdouble{$\redinfNname{n}$-sleep}{
        \redinfN n {C;S;E}
}{
        \redinfN {n+1} {C;S;E}
}
\irulenumberdouble{$\redinfNname{n}$-perform}{
        \redplus {C;S;E} {C';S';E'} \and
        \redinfN {n'} {C';S';E'}
}{
        \redinfN {n} {C;S;E}
}
\end{ruleset}
The relation $\redinfNname{n}$ is similar to $\redinfname$, but
allows the abstract machine to remain in the same state, not
performing any transitions, for at most $n$ steps (rule
$\redinfNname{n}$-sleep).
If $n$ drops to zero, one or several transitions must be performed
(rule $\redinfNname{n}$-perform).
In exchange for performing at least one transition, the
count $n$ can be reset to any value $n'$, allowing an arbitrary but
finite number of non-transitions to be taken afterwards.

A proof by coinduction shows the following variant of
theorem~\ref{compile-evalinf}, using $\redinfNname{n}$ with $n$ equal to
the left application height of the term under consideration.

\begin{lemma} \label{compile-evalinf-aux}
If $e \vdash \evalinf a$, then
$ \redinfN {\leftappheight a} {(\compile{a}, C) ;~ S ;~ \compile{e}} $
\end{lemma}

\begin{constrproof}  By coinduction and case analysis on the last rule
used to derive $e \vdash \evalinf a$.  In the first case, $a = a_1~a_2$
and $e \vdash \evalinf {a_1}$.  Applying the coinduction hypothesis, we obtain
$
\redinfN {\leftappheight {a_1}}
         {(\compile{a_1}, \compile{a_2}, {\tt App}, C) ;~ S ;~ \compile{e}} $
and the result follows by one application of rule
($\redinfNname{n}$-sleep), noticing that 
$\leftappheight{a} = \leftappheight{a_1} + 1$.

In the second case, $a = a_1~a_2$, $e \vdash \eval {a_1} v$ 
and $e \vdash \evalinf {a_2}$.  By lemma~\ref{compile-eval}, we obtain
$ \redplus {(\compile{a_1}, \compile{a_2}, {\tt App}, C) ;~ S ;~ \compile{e}}
            {(\compile{a_2}, {\tt App}, C) ;~ \compile{v_1}.S ;~ \compile{e}}
$.
Using the coinduction hypothesis, we also have
$
\redinfN {\leftappheight {a_2}}
         {(\compile{a_2}, {\tt App}, C) ;~ \compile{v_1}.S ;~ \compile{e}}
$.
The result follows by rule ($\redinfNname{n}$-perform).
The third case of divergence is similar and we omit it.
\end{constrproof}

We then show the following implication between
$\redinfNname{n}$ and $\redinfname$.

\begin{lemma} \label{redinfN-redinf}
If $\redinfN{n}{C;S;E}$, then $\redinf{C;S;E}$.
\end{lemma}

\begin{constrproof}  We first show that
$\redinfN{n}{C;S;E}$ implies
the existence of $n'$, $C'$, $S'$ and $E'$ such that
$\redone {C;S;E} {C';S';E'}$ and $\redinfN{n'}{C';S';E'}$
by Peano induction over $n$.  The result then follows by 
coinduction. 
\end{constrproof}

Theorem~\ref{compile-evalinf} then follows from lemmas
\ref{compile-evalinf-aux}~and~\ref{redinfN-redinf}.
We omit the proof of theorem~\ref{compile-coeval}, which is similar.

\section{Related work}

There are few instances of coinductive definitions and proofs for
big-step semantics in the literature.  Cousot and Cousot \cite{Cousot-92}
proposed the coinductive big-step characterization of divergence that
we use in this article and studied its applicability for abstract
interpretation, as pursued later by Schmidt \cite{Schmidt98}.
This approach was applied to call-by-name
$\lambda$-calculus by Hughes and Moran \cite{Hughes-Moran-95} and by
Crole \cite{Crole-coind-98}, and to call-by-value
$\lambda$-calculus by Grall \cite{Grall-phd}.

Following up on \cite{Cousot-92}, Cousot and Cousot recently
introduced bi-inductive semantics and applied it to the call-by-value
$\lambda$-calculus \cite{Cousot-bi-ind-sos07}.  Bi-inductive
semantics are defined in terms of smallest fixed points with respect to a
nonstandard ordering.  This approach captures both terminating and
diverging executions using a common set of inference rules.  For
instance, in the case of the call-by-value $\lambda$-calculus,
a single inference rule replaces the two rules ($\evalname$-app) and
($\evalinfname$-app-f) of our presentation.  It is not entirely clear
yet how the bi-inductive approach could be mechanized in a proof assistant.
Another difference with the present article is that Cousot and Cousot
\cite{Cousot-bi-ind-sos07} start from a big-step trace semantics,
then systematically derive the other semantics (big-step and
small-step) by abstraction: this is an interesting alternative to our
approach that separately deals with each semantics.

Gunter and R\'emy \cite{Gunter-Remy-ravl} and
Stoughton \cite{Stoughton-98}
have the same initial goal as us, namely describe both terminating and
diverging computations with big-step semantics, but use increasing
sequences of finite, incomplete derivations to do so, instead of
infinite derivations.  We do not know yet how their approach relates
to our $\evalinfname$ and $\coevalname$ relations.

Milner and Tofte \cite{Milner-Tofte-coinduction} and later 
Leroy and Rouaix \cite{Leroy-Rouaix-99} used coinduction in the context of
big-step semantics for functional and imperative languages, not to
describe diverging evaluations, but to capture safety properties over
possibly cyclic memory stores.  

Of course, coinductive techniques are routinely used in the context of
small-step semantics, especially for the labeled transition systems
arising from process calculi.  
The flavours of coinduction used there,
especially proofs by bisimulations, are quite different from the
present work.
These techniques closely resemble the way coinduction can be used  
for defining the contextual equivalence in an operational setting 
\cite{Pitts-equiv-97}
and the approximation order in
the recursively defined domains involved in denotational semantics
\cite{Pitts94}.

The infinitary $\lambda$-calculus
\cite{Kennaway-KSV-97,Berarducci-Dezani-99} studies diverging
computations from a very different angle: not only the authors use
reduction semantics, but their terms are also infinite, and they use
topological techniques (metrics, convergence, etc) instead of coinduction.

\section{Conclusions}

We investigated two coinductive approaches to giving big-step
semantics for non-terminating computations.  The first, based on
\cite{Cousot-92} and using separate evaluation rules for
terminating terms and diverging terms, appears very well-behaved: it
corresponds exactly to finite and infinite reduction sequences, 
and lends itself well to type soundness proofs and to compiler
correctness proofs.  The second approach, consisting in a coinductive
interpretation of the standard evaluation rules, is less satisfactory:
while amenable to compiler correctness proofs as well,
it captures only a subset of the diverging computations of interest
--- and it is not yet clear which subset exactly.

To evaluate the applicability of the coinductive techniques presented
here to languages other than small functional languages, we developed
coinductive big-step semantics for three low-level imperative
languages used in the Compcert verified compiler \cite{Leroy-Compcert-Coq}:
the source language Clight (a large subset of the C language) and
the two intermediate languages C\#minor and Cminor.  These semantics
characterize non-terminating programs and the traces of input/output
events they perform.  These semantics were used to mechanically prove
that the first four passes of the Compcert compiler preserve the
semantics of diverging programs.  Some of the proofs use techniques similar to
those presented in section~\ref{s-semantic-preservation} to combine
co-inductive and inductive reasoning.  The results of this
experiment are encouraging.  In particular, the addition of
coinductive rules for divergence increases the size of the semantics
by 40\% only.

\section*{Acknowledgments}

Andrzej Filinski disproved the conjecture from
section~\ref{s-soundness-bigstep} very shortly after it was stated.  We
thank Eduardo Bonelli, the anonymous reviewers for the ESOP 2006 conference,
the participants of the 22nd meeting of IFIP Working Group 2.8
(Functional Programming), and the anonymous reviewers of this special
issue for their feedback.

\bibliographystyle{elsart-num}
\bibliography{biblio}

\end{document}